\documentclass{article}

\usepackage{fullpage}
\usepackage{url}
\usepackage{xspace}

\usepackage{graphics}
\usepackage[dvips]{epsfig}

\usepackage{amsmath,amsthm}
\usepackage{amssymb}
\usepackage{amsfonts}
\usepackage{graphicx}
\usepackage{algorithm,algorithmicx}
\usepackage{algpseudocode}

\newtheorem{theorem}{Theorem}[section]
\newtheorem{lemma}{Lemma}[section]

\newtheorem{claim}{Claim}[section]
\newtheorem{proposition}{Proposition}[section]

\newtheorem{fact}{Fact}[section]

\newcommand{\cA}{{\cal A}}

\algnewcommand{\LineComment}[1]{\Statex #1}

\newcommand{\remove}[1]{}

\begin{document}

\baselineskip  0.2in 
\parskip     0.1in 
\parindent   0.0in 

\title{{\bf Tradeoffs Between Cost and Information\\ for Rendezvous and Treasure Hunt }}
\date{}
\newcommand{\inst}[1]{$^{#1}$}

\author{
Avery Miller\inst{1},
Andrzej Pelc\inst{1}$^,$\footnote{Partially supported by NSERC discovery grant and by the Research Chair in Distributed Computing at the Universit\'e du Qu\'{e}bec en Outaouais.}\\
\inst{1} Universit\'{e} du Qu\'{e}bec en Outaouais, Gatineau, Canada.\\
E-mails: \url{avery@averymiller.ca}, \url{pelc@uqo.ca}\\
}

\maketitle

\begin{abstract}

In rendezvous, two agents traverse network edges in synchronous rounds and have to meet at some node. In treasure hunt, a single agent has to find a stationary target situated at an unknown node of the network. We study tradeoffs between the amount of information ({\em advice}) available  {\em a priori} to the agents and the cost (number of edge traversals) of rendezvous and treasure hunt. Our goal is to find the smallest size of advice which enables the agents to solve these tasks at some cost $C$ in a network with $e$ edges. This size turns out to depend on the initial distance $D$ and on the ratio $\frac{e}{C}$, which is the {\em relative cost gain} due to advice. For arbitrary graphs, we give upper and lower bounds of $O(D\log(D\cdot \frac{e}{C}) +\log\log e)$ and $\Omega(D\log \frac{e}{C})$, respectively, on the optimal  size of advice. For the class of trees, we give nearly tight upper and lower bounds of $O(D\log \frac{e}{C} + \log\log e)$ and $\Omega (D\log \frac{e}{C})$, respectively.

\vspace{1ex}


\end{abstract}

\maketitle

\section{Introduction}

\subsection{Model and problems}

Rendezvous and treasure hunt are two basic tasks performed by mobile agents in networks. In rendezvous, two agents, initially located at distinct nodes of the network,
traverse network edges in synchronous rounds and
have to meet at some node. In treasure hunt, a single agent has to find a stationary target (called treasure) situated at an unknown node of the network. 
 The network might model a labyrinth or a system of corridors in a cave, in which case the agents might be mobile robots. The meeting of such robots might be motivated by
 the need to exchange previously collected samples, or to agree how to share a future cleaning or decontamination task. Treasure hunt might mean
  searching a cave for a resource or for a missing person after an accident.
 In other applications we can consider a computer network, in which the mobile entities are software agents.
The meeting of such agents might be necessary to exchange data or share a future task of checking the functionality of network components. Treasure hunt in this case might mean  
  looking for valuable data residing at some node of the network, or for a virus implanted
at some site.

The network is modeled as a simple undirected connected graph whose nodes have distinct identities. Ports at a node of degree $d$ are numbered $0, \dots,d-1$. The agents are anonymous, i.e., do not have identifiers. 
Agents execute a deterministic algorithm, such that, at each step, they choose a port at the current node.
When an agent enters a node, it learns the entry port number, the label of the node and its degree.  
The cost of a rendezvous algorithm is the total worst-case number
of edge traversals performed by both agents until meeting. 
The cost of a treasure hunt algorithm is the
worst-case number
of edge traversals performed by the agent until the treasure is found.
If the agents have no information about the network, the cost of both rendezvous and treasure hunt
can be as large as $\Theta(e)$ for networks with $e$ edges. This is clear for treasure hunt, as all edges (except one) need to be traversed by the agent to
find the treasure in the worst case. The same lower bound for rendezvous follows from Proposition \ref{eq} in the present paper. On the other hand, if 
$D$ is the distance between the initial positions of the agents, or from the initial position of the agent to the treasure, a lower bound on the cost of rendezvous and of treasure hunt is $D$.

In this paper, we study tradeoffs between the amount of information available  {\em a priori} to the agents and the cost of rendezvous and treasure hunt. Following the paradigm of algorithms
with advice \cite{AKM01,CFP,CFIKP,DP,EFKR,FGIP,FIP1,FIP2,FKL,FP,FPR,GPPR02,IKP,KKKP02,KKP05,SN,TZ05}, this information is provided to the agents at the start of their navigation by an oracle that knows the network, the starting  positions of the agents and, in the case of treasure hunt, the node where the treasure is hidden. The oracle
assists the agents by providing them with a binary string called {\em advice}, which can be used by the agent during the algorithm execution. 
In the case of rendezvous, the advice given to each agent can be different.
The length of the
string given to the agent in treasure hunt and the sum of the lengths of strings given to both agents in rendezvous is called the {\em size of advice}.

 \subsection{Our results}
 
 Using the framework of advice permits us to quantify the amount of information
needed for an efficient solution of a given network problem (in our case, rendezvous and treasure hunt) regardless of the type of information that is provided. 
 Our goal is to find the smallest size of advice which enables the agents to solve rendezvous and treasure hunt at a given cost $C$ in a network with $e$ edges.
This size turns out to depend
on the initial distance $D$ (between the agents in rendezvous, and between the agent and the treasure in treasure hunt) and on the ratio $\frac{e}{C}$, which is the {\em relative cost gain} due to advice. 
For arbitrary graphs, we give upper and lower bounds of $O(D\log(D\cdot \frac{e}{C}) +\log\log e)$ and $\Omega(D\log \frac{e}{C})$, respectively, on the optimal  size of advice. Hence our bounds leave only a
logarithmic gap in the general case. For the class of trees, we give nearly tight upper and lower bounds of $O(D\log \frac{e}{C} + \log\log e)$ and $\Omega (D\log \frac{e}{C})$, respectively. Our upper bounds are 
obtained by constructing an algorithm for all graphs (respectively, for all trees) that works at the given cost and with 
advice of the given size, while the lower bounds are proved by exhibiting networks for which it is impossible to achieve the given cost 
with  smaller advice. 

\subsection{Related work}

Treasure hunt, network exploration and rendezvous in networks are interrelated problems that have received much attention in recent literature.
Treasure hunt has been investigated in the line \cite{BCD,HIKL}, in the plane \cite{BCR} and in other terrains \cite{LS}. Treasure hunt in anonymous networks
(without any information about the network) has been studied in \cite{TSZ07,Xin} with the goal of minimizing cost.

The related problem of graph exploration by mobile agents (often called robots) has been
intensely studied as well. The goal of this task is to visit all of the nodes and/or traverse all of the edges of a graph. A lot of  research considered the case of a
single agent exploring a labeled graph.  In \cite{AH,DePa} the
agent explores strongly-connected directed graphs. In a directed graph, an agent can move only
in the direction from tail to head of an edge, not vice-versa.  In
particular, \cite{DePa} investigated the minimum time of exploration of
directed graphs, and \cite{AH} gave improved algorithms for this
problem in terms of the deficiency of the graph (i.e., the minimum
number of edges that must be added to make the graph Eulerian).  Many papers,
e.g., \cite{DFKP,DKK,PaPe} studied the scenario where the
graph to be explored is labeled and undirected, and the agent can traverse edges in both
directions.  In
\cite{PaPe}, it was shown that a graph with $n$ nodes and $e$ edges can
be explored in time $e+O(n)$.  In some papers, additional restrictions
on the moves of the agent were imposed, e.g.,   it was assumed that the agent
is tethered,
i.e., attached to the base by a rope or cable of restricted length
\cite{DKK}.
In \cite{Re}, a log-space construction of a deterministic exploration for all graphs with a given bound on size was shown.

The problem of rendezvous has been studied both under randomized and deterministic scenarios.
In the framework of networks, it is usually assumed that the nodes do not have distinct identities.
An extensive survey of  randomized rendezvous in various models  can be found in
\cite{alpern02b}, cf. also  \cite{alpern95a,alpern02a,anderson90,baston98}. 
Deterministic rendezvous in networks has been surveyed in \cite{Pe}.
Several authors
considered geometric scenarios (rendezvous in an interval of the real line, e.g.,  \cite{baston98,baston01},
or in the plane, e.g., \cite{anderson98a,anderson98b}).
Gathering more than two agents was studied, e.g., in \cite{fpsw}.

For the deterministic setting, many authors studied the feasibility and time complexity of rendezvous of synchronous agents, i.e., agents that move in rounds. 
In \cite{MP2} the authors studied tradeoffs between the time of rendezvous and the number of edge traversals by both agents. 
In \cite{DFKP}, the authors presented a rendezvous algorithm whose running time is polynomial in the size of the graph, the length of the shorter
label and the delay between the starting times of the agents. In \cite{KM,TSZ07}, rendezvous time is polynomial in the first two of these parameters and independent of the delay.
The amount of memory required by the agents to achieve deterministic rendezvous was studied in  \cite{CKP} for general graphs.
The amount of memory needed for randomized rendezvous in the ring was discussed, e.g., in~\cite{KKPM08}. 
Several authors investigated asynchronous rendezvous in the plane \cite{CFPS,fpsw} and in network environments
\cite{BCGIL,CLP,DGKKP,DPV}.

Providing nodes or agents with information of arbitrary type that can be used to perform network tasks more efficiently has been
proposed in \cite{AKM01,CFP,CFIKP,DP,EFKR,FGIP,FIP1,FIP2,FKL,FP,FPR,GPPR02,IKP,KKKP02,KKP05,MP,SN,TZ05}. This approach was referred to as
 algorithms with {\em advice}.  
The advice is given either to nodes of the network or to mobile agents performing some network task.
Several of the authors cited above studied the minimum size of advice required to solve the
respective network problem in an efficient way. 

 In \cite{KKP05}, given a distributed representation of a solution for a problem,
the authors investigated the number of bits of communication needed to verify the legality of the represented solution.
In \cite{FIP1}, the authors compared the minimum size of advice required to
solve two information dissemination problems using a linear number of messages. 
In \cite{FKL}, it was shown that a constant amount of advice enables the nodes to carry out the distributed construction of a minimum
spanning tree in logarithmic time. 
In \cite{EFKR}, the advice paradigm was used for online problems.
In \cite{FGIP}, the authors established lower bounds on the size of advice 
needed to beat time $\Theta(\log^*n)$
for 3-coloring a cycle and to achieve time $\Theta(\log^*n)$ for 3-coloring unoriented trees.  
In the case of \cite{SN}, the issue was not efficiency but feasibility: it
was shown that $\Theta(n\log n)$ is the minimum size of advice
required to perform monotone connected graph clearing.
In \cite{IKP}, the authors studied radio networks for
which it is possible to perform centralized broadcasting with advice in constant time. They proved that
$O(n)$ bits of advice allow to obtain constant time in such networks, while
$o(n)$ bits are not enough. In \cite{FPR}, the authors studied the problem of topology recognition with advice given to nodes. 
In \cite{DP}, the authors considered the task of drawing an isomorphic map by an agent in a graph, and their goal was to determine the minimum amount of advice that has to be given to the agent
for the task to be feasible.

Among the papers using the paradigm of advice, \cite{CFIKP,FIP2, MP} are closest to the present work.  Both  \cite{CFIKP,FIP2} concerned the task of graph exploration by an agent.
In \cite{CFIKP}, the authors investigated the minimum size of advice that has to be given to unlabeled nodes (and not to the agent)
to permit graph exploration by an agent modeled as a $k$-state automaton.
In \cite{FIP2}, the authors
established the size of advice that has to be given to an agent completing exploration of trees, in order to break competitive ratio 2. In \cite{MP}, the authors
studied the minimum size of advice that must be provided to labeled agents, in order to achieve rendezvous at minimum possible cost, i.e., at cost $\Theta(D)$, where $D$ is the initial distance between the agents. They showed that this optimal size of advice for rendezvous in $n$-node networks is $\Theta(D\log(n/D)+\log\log L)$, where
the labels of agents are drawn from the set $\{1,\dots ,L\}$.  This paper differs from the present one in two important aspects. First, 
as opposed to the present paper, in \cite{MP}, agents get identical advice, and nodes of the network are unlabeled.
 Second, instead of looking at tradeoffs between cost and the size of advice, as we do in the present paper, the focus of \cite{MP} was on the size of advice sufficient to achieve the lowest possible cost.

\section{Preliminaries}

In this section we show that, in the context of advice,  treasure hunt and rendezvous are essentially equivalent.
More precisely, the following proposition shows that the minimum advice sufficient to solve both problems at a given cost in the class of graphs with $\Theta(e)$ edges
and with the initial distance $\Theta(D)$ is the same, up to constant factors. Throughout the paper a {\em graph} means a simple connected undirected graph.
The number of nodes in the graph is denoted by $n$, and the number of edges is denoted by $e$.
All logarithms are to base 2.

\begin{proposition}\label{eq}
Let  $D \leq e$ be positive integers.  
\begin{enumerate}
\item
If there exists an algorithm {\tt TH} that solves treasure hunt at cost $C$ with advice of size $A$  in all graphs with $e$ edges and with initial distance $D$ between  the
agent and the treasure, then there exists an algorithm {\tt RV} that solves rendezvous at cost $C$ with advice of size $A+2$
in all graphs with $e$ edges and with initial distance $D$ between  the
agents.
\item
If there exists an algorithm {\tt RV} solving rendezvous at cost $C$ with advice of size less than $A$ in all graphs with $2e+1$ edges and with initial distance $2D+1$ between the agents, then there exists an algorithm {\tt TH}
that solves treasure hunt at cost at most $C$ with advice of size at most  $A$ in all graphs with $e$ edges and with initial distance $D$ between the agent and the treasure.
\end{enumerate}
\end{proposition}

\begin{proof}
{\em Part 1.}
Consider a graph $G$ with $e$ edges, and two agents, $a$ and $b$, that have to meet. Suppose that $a$ and $b$ start at nodes $v$ and $w$ in graph $G$, and that $D$ is the distance between $v$ and $w$.
Let $\alpha$ be the advice string of size $A$ that enables an agent starting at $v$ to find the treasure located at $w$ at cost $C$ using 
algorithm $\tt TH$. Give advice string $(0)$ to agent $b$
and advice string $(1\alpha)$ to agent $a$. The sum of the lengths of these strings is $A+2$. 
The rendezvous algorithm {\tt RV} is the following. With advice string $(0)$ stay inert; with advice string $(1\alpha)$ execute algorithm
$\tt TH$ using advice $\alpha$. By the correctness of {\tt TH}, this rendezvous algorithm is correct and its cost is $C$. 

{\em Part 2.}
Consider a graph $G$ with $e$ edges and with initial distance $D$ between the agent (initially located at $v$) and the treasure (initially located at $w$).
We construct the following graph $G'$. It consists of two disjoint copies $H_0, H_1$ of $G$ with the respective nodes $w$ in each copy joined by an additional edge $f$.
The graph $G'$ has $2e+1$ edges. Label nodes of the graph $G'$ as follows. If some node of $G$ has label $\ell$, then the corresponding node in $H_0$ has
label $2\ell$ and the corresponding node in $H_1$ has
label $2\ell +1$.  
Place two agents in $G'$, each at the node $v$ of a different copy of graph $G$. Hence, the initial positions of the agents are
at distance $2D+1$ in $G'$. Let $\alpha _0$ and $\alpha _1$ be the advice strings (whose lengths sum to less than $A$) that are provided to the agents
starting in $H_0$ and $H_1$, respectively, in the execution of {\tt RV}
in $G'$.
 In this execution, at least one of the agents has to traverse edge $f$, and, hence, it has to reach the node $w$
in its copy $H_i$ of $G$. 
Therefore it travels from $v$ to $w$ in $H_i$ with an advice string $\alpha_i$ of size less than $A$, at cost at most $C$. Algorithm {\tt TH}
for treasure hunt in $G$  is given the advice string $\alpha_i$ with the single bit $i$ appended. The algorithm consists of the solo execution of {\tt RV}
where the agent transforms the label $\ell$ of each visited node to $2\ell +i$.
\end{proof}

In view of Proposition \ref{eq}, in the rest of the paper we can restrict attention to the problem of treasure hunt. All of our results, both the upper and the lower bounds,
also apply to the rendezvous problem (with the provision that, if treasure hunt can be solved at cost $C$ with no advice, then rendezvous can be solved at cost $C$ with constant advice). Notice that the equivalence of rendezvous and treasure hunt depends on the fact that, in rendezvous,  the oracle can give
different pieces of advice to the two agents. If the oracle was forced to give the same advice to both agents, then symmetry could not be broken in all cases since agents are anonymous, and rendezvous would be impossible in some networks.

\section{Treasure Hunt in Arbitrary Graphs}
In this section, we proceed to prove upper and lower bounds on the advice needed to solve treasure hunt in arbitrary graphs. These bounds are expressed in terms of $D$, which is the distance between the treasure and the initial position of the agent, and in terms of the ratio $\frac{e}{C}$, where $e$ is the number of edges in the graph and $C$ is an upper bound on the cost of the algorithm. This ratio is the relative cost gain due to advice. 
We first provide an algorithm that solves treasure hunt using $O(D\log (D\cdot \frac{e}{C})+\log\log e)$ bits of advice, and then prove that any deterministic algorithm for this task uses at least $\Omega(D\log \frac{e}{C})$ bits of advice.

\subsection{Algorithm}\label{FindTreasure}
Consider an $n$-node graph $G$ and a node $s$ of $G$, which is the initial position of the agent. Let $P = (v_0,\ldots,v_D)$ be a shortest path from $s$ to the treasure, where $v_i$ is the node at distance $i$ from $s$ along path $P$. Let $\mathit{LogSum} = \sum_{i=0}^{D-1} \lceil \log(\mathit{deg}(v_i)) \rceil$. Intuitively, $\mathit{LogSum}$ is an upper bound on the total number of bits needed to fully describe the sequence of ports leading from $s$ to the treasure. For any fixed integer $\ell \in \{1,\ldots,\mathit{LogSum}\}$, we describe a binary advice string of length $O(\ell + \log D + \log\log{e})$ and an algorithm that uses this advice when searching for the treasure. We do not consider values of $\ell$ greater than $\mathit{LogSum}$ since we will show that, when $\ell = \mathit{LogSum}$, our algorithm has optimal cost $D$.

To construct the advice, the idea is to use $\ell$ bits to produce $D$ advice substrings to guide the agent along path $P$. In particular, the first $\ell$ bits of advice consist of $D$ binary substrings $A_0,\ldots,A_{D-1}$. For each $i \in \{0,\ldots,D-1\}$, the substring $A_i$ is created by considering the node $v_i$ on path $P$ that is at distance $i$ from $s$ in $G$. The length of $A_i$ is dictated by the ratio of the number of bits needed to describe the degree of $v_i$ to the total number of bits needed to describe the degrees of all nodes on path $P$. The set of ports at $v_i$ is partitioned into numbered  \emph{sectors}  (i.e., subintervals) of size at most $\lceil \mathit{deg}(v_i)/2^{|A_i|} \rceil$. In fact, at most one of the sectors can have size smaller than this value. 
The substring $A_i$ is taken to be the binary representation of the number of the sector containing the port that leads to the next node $v_{i+1}$ on path $P$ towards the treasure. 

Below, we provide pseudocode that describes how the advice is created. First, Algorithm \ref{createadvice} finds a shortest path $P$ from $s$ to the treasure. The path consists of node/port pairs $(v_i,p_i)$ for each $i \in \{0,\ldots,D-1\}$, where $v_0 = s$ and, for each $i \in \{0,\ldots,D-1\}$, port $p_i$ leads from node $v_i$ to node $v_{i+1}$. The sum $\sum_{i=0}^{D-1} \lceil \log(\mathit{deg}(v_i)) \rceil$ is calculated and stored in $\mathit{LogSum}$. For ease of notation, we define $\beta = \ell / \mathit{LogSum}$. Each pair $(v_i,p_i)$ is passed to the subroutine described in Algorithm \ref{encodesector}, along with $\beta$. This subroutine uses $\beta$ and the degree of $v_i$ to determine the appropriate number $z_i$ of advice bits via the formula $z_i =  \left\lfloor \lceil\log{(\mathit{deg}(v)})\rceil \cdot \beta \right\rfloor$, then divides the set of ports at $v_i$ into numbered sectors, determines to which sector port $p_i$ belongs, and outputs the binary representation of this sector number as a $z_i$-bit string $A_i$.

The resulting sequence of substrings $(A_0,\ldots,A_{D-1})$, along with the binary string $LS$ representing the value of $\mathit{LogSum}$, is encoded into a single advice string to pass to the algorithm. More specifically, these strings are encoded by doubling each digit in each substring and putting 01 between substrings. This permits the agent to unambiguously decode the original sequence, to calculate the value of $D$ by looking at the number of separators 01, and to calculate the value of $\ell$ by looking at the lengths of the first $D$ advice substrings. Denote by $Concat(A_0,\ldots,A_{D-1},LS)$ this encoding and let $Decode$ be the inverse (decoding) function, i.e. $Decode(Concat(A_0,\ldots,A_{D-1},LS)) = (A_0,\ldots,A_{D-1},LS)$. As an example, $Concat((01),(00)) = (0011010000)$. Note that the encoding increases the total number of advice bits by a constant factor. The advice string, calculated by Algorithm  \ref{createadvice} using the strings $A_i$ supplied by Algorithm \ref{encodesector},
 is ${\cal A}= Concat(A_0,\ldots,A_{D-1},LS)$. The advice string $\cA$ is given to the agent.

\begin{algorithm}[H]
\caption{\texttt{CreateAdvice}($G$,$s$,$\ell$)}
\begin{algorithmic}[1]
\State Find a shortest path $P=\{v_0,\ldots,v_{D-1},v_D\}$ in $G$ from node $s$ to the node containing the treasure.
\State $\mathit{LogSum} \leftarrow \sum_{i=0}^{D-1} \lceil \log(\mathit{deg}(v_i)) \rceil$
\State $\beta \leftarrow \ell / \mathit{LogSum}$
\For{$i=0,\ldots,D-1$}
\State $p_i \leftarrow$ port number leading from $v_i$ to node on path $P$ at distance $i+1$ from $s$
\State $A_i \leftarrow $ \texttt{EncodeSectorNumber}$(v_i,p_i,\beta)$
\EndFor
\State $LS \leftarrow$ binary representation of $\mathit{LogSum}$
\State Output $Concat(A_0,\ldots,A_{D-1},LS)$
\end{algorithmic}
\label{createadvice}
\end{algorithm}

\begin{algorithm}[H]
\caption{\texttt{EncodeSectorNumber}$(v,\mathit{port},\beta)$}
\begin{algorithmic}[1]
\State $z \leftarrow \left\lfloor \lceil\log{(\mathit{deg}(v)})\rceil \cdot \beta\right\rfloor$
\State $\mathit{SectorSize} \leftarrow \lceil \mathit{deg}(v)/2^{z} \rceil$
\State $\mathit{SectorNumber} \leftarrow \lfloor \mathit{port}/\mathit{SectorSize} \rfloor$
\State // \emph{port} is contained in the range $\{\mathit{SectorNumber}\cdot\mathit{SectorSize},\ldots,(\mathit{SectorNumber}+1)\cdot\mathit{SectorSize}-1\}$
\State return the $z$-bit binary representation of $\mathit{SectorNumber}$
\end{algorithmic}
\label{encodesector}
\end{algorithm}

\begin{lemma}\label{advicesize}
The advice string ${\cal A}= Concat(A_0,\ldots,A_{D-1},LS)$ has size $O(\ell + \log D + \log\log{e})$. 
\end {lemma}

\begin{proof}
Each of the strings $A_i$ has length $\left\lfloor \frac{\lceil\log{(\mathit{deg}(v_i)})\rceil}{\mathit{LogSum}} \cdot \ell \right\rfloor$, and the sum of these lengths
is at most $\ell$. The string $LS$ is the binary encoding of the sum $\sum_{i=0}^{D-1} \lceil \log(\mathit{deg}(v_i)) \rceil$. This sum is maximized when all $v_i$ have the same
degree, hence it is  $O(D( \log (e/D)+1))$. It follows that the length of $LS$ is $O(\log D + \log\log e)$. Therefore, the length of $\cal A$
is in $O(\ell + \log D + \log\log{e})$.  
\end{proof}

Next, we describe the algorithm {\tt FindTreasure}, which is the agent's algorithm given an advice string $\cA=Concat(A_0,\ldots,A_{D-1},LS)$. For the purpose of description only, we define the {\em trail} of the agent, which is a stack of edges that it has previously traversed. The stack gets popped when the agent backtracks. The agent performs a walk in $G$ starting at node $s$. In each step of the algorithm, the agent chooses an edge to add to the trail, or it backtracks along the trail edge that it added most recently. The number of edges in the agent's trail will be used to measure the agent's progress. In particular, when the agent is located at a node $v$ and there are $i$ edges in the agent's trail, we will say that the agent is at \emph{progress level} $i$. The agent keeps track of its current progress level by maintaining a counter that is incremented when it adds a trail edge and decremented when it backtracks.


The agent maintains a table containing the labels of the nodes that it has visited, and, for each node label, the smallest progress level at which the agent visited the node
so far. When the agent arrives at a node $v$ from a lower progress level and does not find the treasure, it checks if its current progress level $i$ is lower than the progress level stored in the table for node $v$. If this is not the case, or if $i=D$, then the agent backtracks by going back along the edge it just arrived on. Also, the agent backtracks immediately if it sees that the degree of $v$ does not ``match'' the size of $A_i$ in the following sense: using $\ell$, the value of $\mathit{LogSum}$ that is encoded in $LS$, and the degree of $v$, the agent checks if $|A_i|$ is equal to the number of bits that the oracle would have provided if $v$ was indeed on the path from $s$ to the treasure, i.e., if $|A_i| = \left\lfloor\lceil\log{(\mathit{deg}(v)})\rceil \cdot \beta \right\rfloor$. Otherwise, if the agent has determined that it should not backtrack immediately, then it uses the advice substring $A_i$ in the following way: it divides the set of port numbers at $v$ into sectors (i.e., intervals of port numbers) of size $\lceil \mathit{deg}(v)/2^{|A_i|} \rceil$, 
gives numbers to the sectors, and then interprets $A_i$ as the binary representation of an integer that specifies one of these sectors. For each port number in the specified sector, the agent takes the port and arrives at some neighbour $w$ of $v$. The agent terminates if it finds the treasure at node $w$, or, otherwise, repeats the above at node $w$. If, after trying all ports at node $v$ in the specified sector, the treasure has not been found, the agent backtracks.

Note that the advice was created with the goal of `steering' the agent in the right direction, i.e., along path $P$, but we can only guarantee that this will happen when the agent is located at nodes on path $P$. In fact, an even stronger condition must hold: for any node $v$ on path $P$ at distance $i$ from $s$, we can only guarantee that the advice will be helpful if the agent is located at node $v$ at progress level $i$, since this is when the agent reads the advice substring $A_i$. In other words, it is possible that the agent visits a node $v$ on $P$ at the `wrong' progress level, in the sense that it won't use the advice that was created specifically for $v$. This is why it is not sufficient to simply have the agent backtrack whenever it arrives at a previously-visited node, since during its previous visit, it may have used the wrong advice. Moreover, we must ensure that the algorithm gracefully deals with the situation where the agent is at a node $w$ at progress level $j$, but the advice substring $A_j$ specifies ports that do not exist at $w$. In our algorithm, the agent ignores any port numbers that are greater than or equal to the current node's degree.

To summarize, in our algorithm,  the agent searches for the treasure in a depth-first manner, but it cannot perform DFS (even only to distance $D$) because the cost would be too large.
Instead, the agent takes only a fraction of ports at each node, but may possibly have to pay for it by traversing the same edge several times (while in DFS every edge is traversed at most twice). As our analysis will show,
this gives an overall decrease of the total cost, especially when the advice is large.

The pseudocode of the search conducted by algorithm {\tt FindTreasure}
 is described by Algorithm \ref{takestep}.  It shows how the agent takes a step in the graph, i.e., for each $i \in \{0,\ldots,D-1\}$, how it uses $A_i$ to move from a node at progress level $i$ to a node at progress level $i+1$. In order to initiate the search, this algorithm is called at node $s$ with progress level 0
(and $prev=s$).
Algorithm \ref{decodesector}, used as a subroutine in Algorithm \ref{takestep},  shows how the agent decodes substring $A_i$ to obtain a range of port numbers.  We assume that we have two functions related to the agent-maintained table of visited nodes: \texttt{UpdateTable}$(v,i)$ that writes $i$ into the entry for node $v$ as the smallest progress level at which the agent has ever visited node $v$, and \texttt{CurrentMin}$(v)$ that reads the entry of the table for node $v$. Each table entry is initialized to $\infty$.

\begin{algorithm}[H]
\caption{\texttt{TakeStep($\cA$,$v$,$i$,$prev$)}}
\begin{algorithmic}[1]
\LineComment{$\cA$ is the advice string, $v$ is the node where the agent is currently located, $i$ is the current progress level, $prev$ is the node from which the agent arrived}
\If{treasure is located at $v$}
\State Stop
\EndIf

\If{$(i < D)$ AND $(i < \mathtt{CurrentMin}(v))$} \label{line:ifline}


\State \texttt{UpdateTable}$(v,i)$ \label{line:updatetable}

\State $(A_0,\dots A_{D-1},LS) \leftarrow Decode (\cA)$
\State $\ell \leftarrow \sum_{i=0}^{D-1} |A_i|$
\State $\mathit{LogSum} \leftarrow$ integer value encoded in binary string $LS$
\State $\beta \leftarrow \ell / \mathit{LogSum}$
\If{ $|A_i| = \left\lfloor \lceil\log{(\mathit{deg}(v)})\rceil \cdot \beta \right\rfloor$ } \label{line:sizematch}
\State $\mathit{sector} \leftarrow \mathtt{GetSector}(v,A_i)$
\For{each port $p$ in $\mathit{sector}$} \label{line:forloop}
		\If{$p < deg(v)$}
		\State take port $p$ \label{line:takeport}
         \State $w \leftarrow$ the node reached after taking port $p$
		\State  call $\mathtt{TakeStep}(\cA,w,i+1,v)$
		\EndIf
\EndFor

\EndIf
\EndIf

\State Return to node $prev$ \label{line:backtrack}

\end{algorithmic}
\label{takestep}
\end{algorithm}

\begin{algorithm}[H]
\caption{\texttt{GetSector}$(v,\mathit{SectorNumberEncoding})$}
\begin{algorithmic}[1]
\State $\mathit{z} \leftarrow $ number of bits in \textit{SectorNumberEncoding}
\State $\mathit{SectorSize} \leftarrow \left\lceil \frac{\mathit{deg}(v)}{2^{z}} \right\rceil$
\State $\mathit{SectorNumber} \leftarrow $ integer value of \textit{SectorNumberEncoding}
\State return $\{\mathit{SectorNumber}\cdot\mathit{SectorSize},\ldots,(\mathit{SectorNumber}+1)\cdot\mathit{SectorSize}-1\}$
\end{algorithmic}
\label{decodesector}
\end{algorithm}

\subsection{Analysis}

In what follows, let $P$ be the path from $s$ to the treasure that is used to create the advice string $\cA = Concat(A_0,\ldots,A_{D-1},LS)$. Suppose that $P$ consists of the nodes $v_0,\ldots,v_D$, where, for each $i \in \{0,\ldots,D\}$, $v_i$ is at distance $i$ from $s$, and the treasure is located at node $v_D$. Also, for each $i \in \{0,\ldots,D-1\}$, let $p_i$ be the port at node $v_i$ that leads to node $v_{i+1}$.

To prove the correctness of the algorithm, we first consider an arbitrary node $v_i$ on path $P$ and suppose that the agent is at progress level $i$. Clearly, this occurs at least once during the execution of \texttt{FindTreasure} since the agent is initially located at $v_0$ at progress level 0. One of the ports at $v_i$ that are specified by the advice substring $A_i$ leads to node $v_{i+1}$, but the agent may try some other of these ports first. We show that either the agent finds the treasure by recursively calling \texttt{TakeStep} after taking one of these other ports, or, the agent eventually takes the port that leads to node~$v_{i+1}$.

\begin{lemma}\label{makesprogress}
For any $i \in \{0,\ldots,D-1\}$, consider the first time that the agent is located at node $v_i$ at progress level $i$. During the execution of \texttt{TakeStep}$(\cA,v_i,i,w)$,
for some node $w$, either:
\begin{enumerate}
\item the agent moves to node $v_{i+1}$ at progress level $i+1$, or,
\item there is a node $v \neq v_{i+1}$ such that the agent moves to node $v$ at progress level $i+1$, calls \texttt{TakeStep}$(\cA,v,i+1,v_i)$, and, during its execution, the treasure is found by the agent.
\end{enumerate}
\end{lemma}
\begin{proof}
Since we are considering the agent's first visit to node $v_i$ at progress level $i$, and it is not possible for the agent to visit $v_i$ at a progress level less than $i$, it follows that $\mathtt{CurrentMin}(v_i) > i$. So, the {\bf if} condition on line \ref{line:ifline} evaluates to true. Further, since node $v_i$ was used in the creation of the advice substring $A_i$, it follows that $|A_i| = \left\lfloor \lceil\log{(\mathit{deg}(v_i)})\rceil \cdot \beta \right\rfloor$, so the {\bf if} condition on line \ref{line:sizematch} evaluates to true.
Suppose that the treasure is not found during any execution of \texttt{TakeStep}$(\cA,v,i+1,v_i)$ with $v \neq v_{i+1}$. By the choice of $A_i$, port $p_i$ is located in the range of port numbers returned by \texttt{GetSector}. Since taking port $p_i$ at node $v_i$ leads to node $v_{i+1}$, there exists an iteration of the loop in \texttt{TakeStep} such that the agent moves to node $v_{i+1}$ and increments its progress level to $i+1$.
\end{proof}

Using induction, we extend Lemma \ref{makesprogress} to show that the agent eventually reaches node $v_D$.

\begin{lemma}\label{eventuallyfinds}
For any $i \in \{0,\ldots,D-1\}$,
consider the first time that the agent is located at node $v_i$ at progress level $i$. During the execution of \texttt{TakeStep}$(\cA,v_i,i,w)$,
for some node $w$, the agent finds the treasure.
\end{lemma}
\begin{proof}
The proof proceeds by induction on $D-i$. In the base case, $D=i$, and the agent finds the treasure when it is first located at node $v_D$ at progress level $D$ since the treasure is located at $v_D$. 
As induction hypothesis, assume that, for some $D-i \in \{0,\ldots,D-1\}$, when the agent is first located at node $v_i$ at progress level $i$, the agent finds the treasure during the execution of \texttt{TakeStep}. Now, consider the first time that the agent is located at node $v_{i-1}$ at progress level $i-1$. Note that, by the induction hypothesis, the agent was not previously located at node $v_i$ at progress level $i$, since otherwise, during the execution of \texttt{TakeStep} at the first such visit, the agent would have found the treasure and terminated. 

By Lemma \ref{makesprogress}, when the agent is first located at node $v_{i-1}$ at progress level $i-1$, either:
\begin{enumerate}
\item the agent moves to node $v_i$ at progress level $i$, or, 
\item there is a node $v \neq v_{i}$ such that the agent moves to node $v$ at progress level $i$, calls \texttt{TakeStep}$(\cA,v,i,v_{i-1})$, and, during its execution, the treasure is found by the agent.
\end{enumerate}
In the first case, the induction hypothesis implies that the agent finds the treasure. In the second case, the treasure is found by the agent, so we are done. 
\end{proof}

By Lemma \ref{eventuallyfinds} with $i=0$, the agent finds the treasure during the first execution of {\tt TakeStep}, hence $\mathtt{FindTreasure}$ is correct.
%
%
Next, we consider the cost of algorithm \texttt{FindTreasure}. Our analysis considers the cases $\ell = \mathit{LogSum}$ and $\ell < \mathit{LogSum}$ separately. We proceed to find upper bounds on the cost of algorithm \texttt{FindTreasure} in terms of a fixed upper bound on the amount of advice provided. To prove the upper bounds, we first give upper bounds on the size of the sector returned by \texttt{GetSector}. 

In the first case, we show that when $\ell = \mathit{LogSum}$ (i.e. $\beta=1$) the cost of algorithm \texttt{FindTreasure} is optimal.

\begin{lemma}\label{sectorupperone}
Suppose that $\beta = 1$. For all $i \in \{0,\ldots,D-1\}$, if the agent is located at node $v_i$ at progress level $i$, then the size of the sector returned by $\mathtt{GetSector}(v_i,A_i)$ is exactly 1.
\end{lemma}
\begin{proof}
By the advice construction, $|A_i| = \lfloor \lceil \log (\mathit{deg}(v_i)) \rceil \cdot \beta \rfloor$. Since $\beta = 1$, it follows that $|A_i| = \lceil \log (\mathit{deg}(v_i)) \rceil$. Hence, in the execution of $\mathtt{GetSector}(v_i,A_i)$, the value of $\mathit{SectorSize}$ is a positive integer $\lceil \mathit{deg}(v_i) / 2^{|A_i|} \rceil = \lceil \mathit{deg}(v_i) / 2^{\lceil \log (\mathit{deg}(v_i)) \rceil} \rceil \leq \mathit{deg}(v_i) / \mathit{deg}(v_i) = 1$, as required.
\end{proof}

\begin{lemma}\label{findtreasurecostone}
Suppose that $\beta = 1$. When provided with advice $Concat(A_0,\ldots,A_{D-1},LS)$, the algorithm $\mathtt{FindTreasure}$ has cost $D$.
\end{lemma}
\begin{proof}
By Lemma \ref{sectorupperone}, for each $i \in \{0,\ldots,D-1\}$, when the agent is located at node $v_i$ at progress level $i$, the execution of $\mathtt{GetSector}(v_i,A_i)$ returns exactly 1 port number $p$ leading to node $v_{i+1}$. Since the agent starts at node $v_0$ at progress level 0, it follows that the agent takes exactly $D$ steps to find the treasure. Therefore, when $\beta = 1$, algorithm $\mathtt{FindTreasure}$ has cost exactly $D$.
\end{proof} 

In the second case, we assume that $\ell < \mathit{LogSum}$ (i.e. $\beta < 1$).

\begin{lemma}\label{sectorupper}
Suppose that $\beta < 1$. For all $i \in \{0,\ldots,D-1\}$, the size of the sector returned by $\mathtt{GetSector}(v,A_i)$ is at most $\frac{2\mathit{deg}(v)}{2^{|A_i|}}$.
\end{lemma}
\begin{proof}
Note that, when $\mathtt{GetSector}(v,A_i)$ is executed, it must be the case that line \ref{line:sizematch} evaluated to true, i.e., that $|A_i| = \left\lfloor \lceil\log{(\mathit{deg}(v)})\rceil \cdot \beta \right\rfloor$. In the execution of $\mathtt{GetSector}(v,A_i)$, the variable $\mathit{SectorSize}$ is assigned the value $\left\lceil \frac{\mathit{deg}(v)}{2^{|A_i|}} \right\rceil$. Note that $$\frac{\mathit{deg}(v)}{2^{|A_i|}} = \frac{\mathit{deg}(v)}{2^{ \left\lfloor \lceil\log{(\mathit{deg}(v)})\rceil \cdot \beta \right\rfloor}} \geq \frac{\mathit{deg}(v)}{2^{\lceil\log{(\mathit{deg}(v)})\rceil \cdot \beta}} \geq \frac{\mathit{deg}(v)}{2^{ (\log{(\mathit{deg}(v)})+1)\cdot \beta}} = \frac{\mathit{deg}(v)}{(\mathit{deg}(v))^{ \beta}\cdot 2^{\beta}} = \frac{2^{1/\beta}\mathit{deg}(v)}{(\mathit{deg}(v))^{ \beta}}.$$ Since $\beta < 1$, it follows that $2^{1/\beta} > 1$ and $(\mathit{deg}(v))^{\beta} \leq \mathit{deg}(v)$, so $\frac{2^{1/\beta}\mathit{deg}(v)}{(\mathit{deg}(v))^{ \beta}} > 1$. Therefore, we have shown that $\frac{\mathit{deg}(v)}{2^{|A_i|}} > 1$, which implies that $\lceil \frac{\mathit{deg}(v)}{2^{|A_i|}} \rceil \leq \frac{2\mathit{deg}(v)}{2^{|A_i|}}$, as required.
\end{proof}


We are now ready to calculate an upper bound on the cost of algorithm \texttt{FindTreasure}. We denote by $m \in \{0,\ldots,D-1\}$ an index such that $|A_m| = \max_i\{|A_i|\}$.
\begin{lemma}\label{findtreasurecost}
Suppose that $\beta < 1$. When provided with advice $Concat(A_0,\ldots,A_{D-1},LS)$, the algorithm $\mathtt{FindTreasure}$ has cost at most $\frac{16De^{1+\beta}}{2^{|A_m|}}$.
\end{lemma}
\begin{proof}
It suffices to count the total number of times that line \ref{line:takeport} of \texttt{TakeStep} is called and multiply this value by 2. 
This is because the cost incurred by backtracking (i.e., line \ref{line:backtrack} of \texttt{TakeStep}) is at most 1 for each execution of \texttt{TakeStep}, which amounts to an overall multiplicative factor of at most 2. So, we consider the number of times that line \ref{line:takeport} of \texttt{TakeStep} is called at an arbitrary node $v$. The number of times that the \textbf{for} loop at line \ref{line:forloop} is iterated is at most $2\mathit{deg}(v)/2^{|A_i|}$ when $v$ is visited at progress level $i$, since, by Lemma \ref{sectorupper}, this is an upper bound on the size of the range returned by \texttt{GetSector}. Since line \ref{line:updatetable} ensures that the condition
on line \ref{line:ifline} is true
at most once at each progress level $i \in \{0,\ldots,D-1\}$, it follows that the total number of times that line \ref{line:takeport} is executed is bounded above by $\sum_{i=0}^{D-1} 2\mathit{deg}(v)/2^{|A_i|}$. Taking the sum over all nodes, the total number of calls to \texttt{TakeStep} is bounded above by 
$$\sum_v \sum_{i=0}^{D-1} 2\mathit{deg}(v)/2^{|A_i|}= 2\sum_{i=0}^{D-1} \frac{\sum_v \mathit{deg}(v)}{2^{|A_i|}} \leq 4e\sum_{i=0}^{D-1} \frac{1}{2^{|A_i|}} = \frac{4e}{2^{|A_m|}}\sum_{i=0}^{D-1} 2^{|A_m|- |A_i|} \leq \frac{4e}{2^{|A_m|}}\sum_{i=0}^{D-1} 2^{|A_m|}.$$
Next, since $|A_m| = \lfloor \lceil\log{(\mathit{deg}(v_m)})\rceil \cdot \beta \rfloor \leq (\log{(\mathit{deg}(v_m)})+1)\cdot \beta$, it follows that 
$$\frac{4e}{2^{|A_m|}}\sum_{i=0}^{D-1} 2^{|A_m|} = \frac{4De}{2^{|A_m|}} \cdot 2^{|A_m|} \leq \frac{4De}{2^{|A_m|}} 2^{\log{(\mathit{deg}(v_m)})\cdot \beta}2^{\beta} = \frac{4De}{2^{|A_m|}} (\mathit{deg}(v_m))^{\beta}2^{\beta} \leq  \frac{4De}{2^{|A_m|}} (e)^{\beta}2^{\beta}.$$

Since $\beta < 1$, it follows that $$\frac{4De}{2^{|A_m|}} (e)^{\beta}2^{\beta} < \frac{8De^{1+\beta}}{2^{|A_m|}}.$$

\end{proof}

Finally, we fix an upper bound $C$ on the cost of \texttt{FindTreasure} and re-state Lemmas \ref{findtreasurecostone} and \ref{findtreasurecost} to obtain an upper bound on the amount of advice needed to solve treasure hunt at cost $C$.

\begin{theorem}\label{ub}
Let $G$ be any graph with $e$ edges, and let $3 \leq D \leq e$ be the distance from the initial position of the agent to the treasure. Let $C$ be any integer such that $D \leq C \leq e$. The amount of advice needed to solve treasure hunt at cost at most $C$ is at most $O(D\log (D\cdot \frac{e}{C}) + \log\log e)$ bits.
\end{theorem}
\begin{proof}
First, consider the case where $\beta = 1$. In this case, $C=D$ (by Lemma \ref{findtreasurecostone}) and $\ell = \mathit{LogSum} \in O(D( \log (e/D)+1)) \subseteq O(D\log\frac{De}{C})$. Next, consider the case where $\beta < 1$. By Lemma \ref{findtreasurecost}, Algorithm \texttt{FindTreasure} solves treasure hunt with cost $C \leq \frac{16De^{1+\beta}}{2^{A_m|}}$. It follows that $2^{|A_m|}\leq \frac{16De^{1+\beta}}{C}$, so $|A_m| \leq  \log\left(\frac{16De^{1+\beta}}{C}\right)$. Since $|A_m| \geq |A_i|$ for each $i \in \{0,\ldots,D-1\}$, it follows that $\ell = |A_0| + \cdots + |A_{D-1}| \leq D\log\left(\frac{16De^{1+\beta}}{C}\right) \in O(D\log\frac{De}{C})$. Therefore, regardless of the value of $\beta$, we have shown that $\ell \in O(D\log(\frac{De}{C}))$. By Lemma \ref{advicesize}, the size of advice is $O(\ell + \log{D} + \log\log e) = O(D\log\frac{De}{C} + \log\log e)$.
\end{proof}

\subsection{Lower Bound}\label{advicelower}

The following lower bound follows immediately from Theorem \ref{tree-lb}, which is proven by constructing a tree for which treasure hunt requires $\Omega(D\log \frac{e}{C})$ bits of advice.  
This theorem will be proven in Section 4.

\begin{theorem}\label{lb}
Let $D \leq C \leq e$. There exists a graph $G$ with $\Theta(e)$ edges, and a position of the treasure at distance $D$ from the initial position of the agent, such that treasure hunt at cost $C$ requires $\Omega(D\log \frac{e}{C})$ bits of advice.
\end{theorem}

The gap between the upper bound given by Theorem \ref{ub} and the lower bound given by Theorem \ref{lb} is at most a factor logarithmic in $D$.
Moreover, it should be noted that 
our bounds differ only by an additive term $O(\log\log e)$ whenever $D$ is polynomial in the gain $\frac{e}{C}$.

\section{Treasure Hunt in Trees}
We now proceed to prove upper and lower bounds on the advice needed to solve treasure hunt in trees. Unlike in the case of arbitrary graphs, where our upper and lower bounds may differ by a logarithmic factor, for trees our bounds differ only by an additive term $O(\log\log e)$. Again, our bounds will be expressed in terms of $D$, which is the distance between the treasure and the initial position of the agent, and in terms of the ratio $\frac{e}{C} = (n-1)/C$, where $e$ is the number of edges in the tree, $n$ is the number of nodes, and $C$ is an upper bound on the cost of the algorithm. Also, for any two nodes $a,b$, we will denote by $d(a,b)$ the distance between $a$ and $b$ in the tree, i.e., the number of edges in the simple path between them.

\subsection{Upper Bound}

To obtain our upper bound, we will use algorithm \texttt{FindTreasure} that was defined and proven correct in Section \ref{FindTreasure} for arbitrary graphs. In this section, we provide an analysis of the algorithm specifically for the case of trees, which gives a strictly better upper bound. We start with the following technical lemma, which shows that, if we take the agent's initial position as the root of the tree, the agent's progress level and the agent's current  depth in the tree (i.e., its current distance from the root) do not differ. Essentially, this is because there is only one simple path from the agent's initial position to each node, and the algorithm ensures that the agent's trail does not contain the same edge multiple times.

\begin{lemma}\label{progresslevel}
Consider algorithm {\tt FindTreasure} executed in any tree.
Suppose that, for some neighbouring nodes $v$ and $prev$, $\mathtt{TakeStep}(\cA,v,i,prev)$ is executed at node $v$. If line \ref{line:ifline} evaluates to true, then
progress level $i = d(s,v)$.
\end{lemma}
\begin{proof}
We proceed by induction on the agent's progress level. In the base case, consider progress level $i=0$. Since the first call to \texttt{TakeStep} has $i=0$, and every subsequent call increments the current progress level,
the agent must be located at node $s$. Next, assume that, for some progress level $i \in \{0,\ldots,D-1\}$ and any neighbouring nodes $v$ and $prev$, in the execution of \texttt{TakeStep}$(\cA,v,i,prev)$, if line \ref{line:ifline} evaluates to true, then $i = d(s,v)$. Now,  for some neighbouring nodes $v'$ and $prev'$, consider the execution of \texttt{TakeStep}$(\cA,v',i+1,prev')$.  \texttt{TakeStep} was executed at node $prev'$ at progress level $i$ and line \ref{line:ifline} of this execution evaluated to true. By the induction hypothesis, it follows that $i = d(prev',s)$. 

Next, consider the value of $d(v',s)$. In a tree, there is only one simple path from $s$ to $v'$ and one simple path from $s$ to $prev'$. Since $v'$ and $prev'$ are neighbours, either $v'$ is on the path from $s$ to $prev'$ (in which case $d(prev',s) = d(v',s)+1$) or $prev'$ is on the path from $s$ to $v'$ (in which case $d(v',s) = d(prev',s)+1$). If line \ref{line:ifline} of the execution of \texttt{TakeStep}$(\cA,v',i+1,prev')$ evaluates to true, then $i+1 < \mathtt{CurrentMin}(v')$, i.e., $v'$ was not previously visited at a progress level less than $i+2$. It follows that $v'$ is not located on the path from $s$ to $prev'$. Therefore, it must be the case that $d(v',s) = d(prev',s)+1$, so $i+1 = d(prev',s)+1 = d(v',s)$.
\end{proof}

Next, we proceed to find an upper bound on the cost of algorithm \texttt{FindTreasure} in trees in terms of a fixed upper bound on the amount of advice provided. The proof is analogous to the proof of Lemma \ref{findtreasurecost}, the main difference being that we do not need to multiply by a factor of $D$ in order to account for the different paths that the agent could use to reach a given node. As before, we denote by $m \in \{0,\ldots,D-1\}$ an index such that $|A_m| = \max_i\{|A_i|\}$.

\begin{lemma}\label{findtreasurecostintrees}
Suppose that $\beta < 1$. When provided with advice $Concat(A_0,\ldots,A_{D-1},LS)$, the algorithm $\mathtt{FindTreasure}$ has cost at most $\frac{16e^{1+\beta}}{2^{|A_m|}}$.
\end{lemma}
\begin{proof}
As in Lemma \ref{findtreasurecost}, it suffices to count the total number of times that line \ref{line:takeport} of \texttt{TakeStep} is called and multiply this value by 2. So, we consider the number of times that line \ref{line:takeport} of \texttt{TakeStep} is called at an arbitrary node $v$. Since line \ref{line:takeport} is only executed if line \ref{line:ifline} evaluates to true, then, by Lemma \ref{progresslevel}, it follows that $i = d(s,v)$ at line \ref{line:takeport}. By Lemma \ref{sectorupper}, the \textbf{for} loop at line \ref{line:forloop} is iterated at most $2\mathit{deg}(v)/2^{|A_{d(s,v)}|}$ times. Taking the sum over all nodes, the total number of calls to \texttt{TakeStep} is bounded above by $$\sum_v \frac{2\mathit{deg}(v)}{2^{|A_{d(s,v)}|}} = \frac{2}{2^{|A_m|}} \sum_v \mathit{deg}(v)\cdot 2^{|A_m|- |A_{d(s,v)}|} \leq \frac{2}{2^{|A_m|}} \sum_v \mathit{deg}(v)\cdot 2^{|A_m|}.$$

Next, since $|A_m| = \lfloor \lceil\log{(\mathit{deg}(v_m)})\rceil \cdot \beta \rfloor \leq (\log{(\mathit{deg}(v_m)})+1)\cdot \beta$, it follows that 
$$\frac{2}{2^{|A_m|}} \sum_v \mathit{deg}(v)\cdot 2^{|A_m|} \leq \frac{2}{2^{|A_m|}} \sum_v \mathit{deg}(v)\cdot 2^{\log{(\mathit{deg}(v_m)})\cdot \beta}2^{\beta} = \frac{2}{2^{|A_m|}} \sum_v \mathit{deg}(v)\cdot (\mathit{deg}(v_m))^{\beta}\cdot 2^{\beta}.$$
Since $\mathit{deg}(v) \leq e$ and $\beta < 1$, it follows that
$$\frac{2}{2^{|A_m|}} \sum_v \mathit{deg}(v)\cdot (\mathit{deg}(v_m))^{\beta}\cdot 2^{\beta} \leq \frac{2^{1+\beta}e^{\beta}}{2^{|A_m|}} \sum_v \mathit{deg}(v) \leq \frac{2^{2+\beta}e^{1+\beta}}{2^{|A_m|}} < \frac{8e^{1+\beta}}{2^{|A_m|}}.$$
\end{proof}

Finally, we fix an upper bound $C$ on the cost of \texttt{FindTreasure} and re-state Lemmas \ref{findtreasurecostone} and \ref{findtreasurecostintrees} as an upper bound on the amount of advice needed to solve treasure hunt in trees at cost $C$.

\begin{theorem}\label{tree-ub}
Let $3 \leq D \leq C \leq e=n-1$. The amount of advice needed to solve treasure hunt on trees of size $n$ with cost at most $C$ is at most $O(D\log \frac{e}{C} + \log\log{e})$ bits.
\end{theorem}
\begin{proof}
First, consider the case where $\beta = 1$. In this case, $C=D$ (by Lemma \ref{findtreasurecostone}) and $\ell = \mathit{LogSum} \in O(D( \log (e/D)+1)) \subseteq O(D\log\frac{e}{C})$. Next, consider the case where $\beta < 1$. By Lemma \ref{findtreasurecostintrees}, Algorithm \texttt{FindTreasure} solves treasure hunt with cost $C \leq \frac{16e^{1+\beta}}{2^{|A_m|}}$. It follows that $2^{|A_m|} \leq \frac{16e^{1+\beta}}{C}$, so $|A_m| \leq \log\left(\frac{16e^{1+\beta}}{C}\right)$. Since $|A_m| \geq |A_i|$ for each $i \in \{0,\ldots,D-1\}$, it follows that $\ell = |A_0| + \cdots + |A_{D-1}| \leq D\log\left(\frac{16e^{1+\beta}}{C}\right) \in O(D\log\frac{e}{C})$. Therefore, regardless of the value of $\beta$, we have shown that $\ell \in O(D\log(\frac{e}{C}))$. By Lemma \ref{advicesize}, the size of advice is $O(\ell + \log{D} + \log\log e) = O(D\log\frac{e}{C} + \log\log e)$.
\end{proof}

\subsection{Lower Bound}

We now set out to prove a lower bound on the amount of advice needed to solve treasure hunt at cost at most $C$. 

We consider a collection $\mathcal{T}(D,k)$ of \emph{caterpillar trees}, each constructed as follows. Take a path graph $P$ consisting of $D+1$ nodes $v_0,\ldots,v_D$, where $v_i$ and $v_{i+1}$ are adjacent, for every $i\in \{0,\dots D-1\}$. Place the treasure at node $v_D$.  
For each $i \in \{0,\ldots,D-1\}$, add $k-1$ nodes to the graph such that each of them has degree 1 and is adjacent only to node $v_i$. The resulting graph is a tree on $Dk+1$ nodes. For each node $v$ in this tree, the ports at $v$ are labeled with the integers $\{0,\ldots,deg(v)-1\}$ so that, for each $i \in \{0,\ldots,D-2\}$, the port numbers at both ends of the edge $\{v_i,v_{i+1}\}$ are equal. Finally we fix node labels as follows.
For each $i \in \{0,\ldots,D-1\}$, node $v_i$ has label $i(k+2)$, and each leaf adjacent to $v_i$ has label $i(k+2)+j+1$, where the port number at $v_i$ leading to it is $j$. Notice that all labels are distinct.

For each $i \in \{0,\ldots,D-1\}$, let $p_i$ be the port number at $v_i$ corresponding to the edge $\{v_i,v_{i+1}\}$. The trees in $\mathcal{T}(D,k)$ are in one-to-one correspondence with the sequences $(p_0,\ldots,p_{D-1})$ because the label of each leaf is determined by the port number (at the adjacent node $v_i$) leading to it. It follows that the number of distinct caterpillar trees in $\mathcal{T}(D,k)$ (taking into consideration the placement of the treasure) is $k^D$. 
Figure \ref{CaterpillarDiagram} gives a diagram of a caterpillar tree in $\mathcal{T}(D,k)$ and shows how nodes are labeled.

\begin{figure}[!ht]
\begin{center}
\includegraphics[scale=0.8]{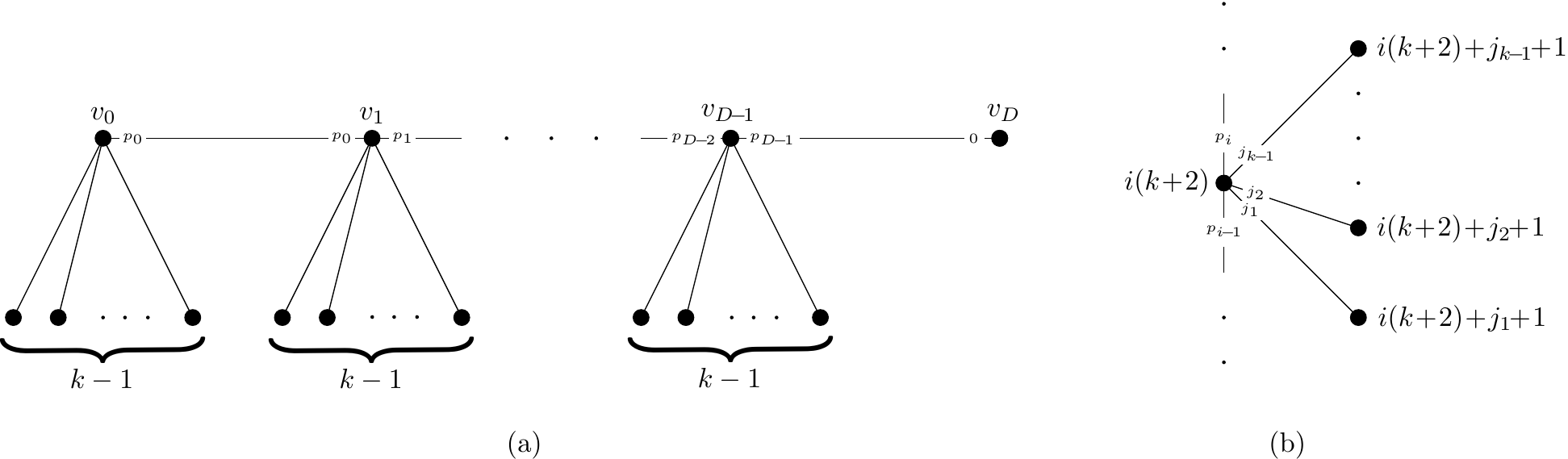}
\end{center}
\caption{(a) A caterpillar tree in $\mathcal{T}(D,k)$ with ports on path $P$ labeled. (b) The labels of the $k-1$ added leaves adjacent to $v_i$ are shown. Node $v_i$ is labeled $i(k+2)$. }
\label{CaterpillarDiagram}
\end{figure}

Consider any fixed caterpillar tree $G \in \mathcal{T}(D,k)$. We set the starting node of the agent to be $v_0$. To find the treasure, the agent must traverse the $D$ edges of path $P$. Suppose that, for some $i \in \{0,\ldots,D-1\}$, the agent is located at node $v_i$. If the agent takes port $p_i$, it will arrive at node $v_{i+1}$, and we say that this edge traversal is \emph{successful}. We may assume that the agent does not return to node $v_i$, i.e., away from the treasure, because such a move would only increase the cost of the algorithm. Further, the agent can detect when it has found the treasure and terminate immediately. 

When an agent's step is not successful (that is, when located at node $v_i$, it chooses a port other than $p_i$) it arrives at a leaf adjacent to $v_i$. In this case, we say that the agent \emph{misses}. After a miss, the agent's next step is to return to node $v_i$. Let $miss_{i,G}$ be the number of times that the agent takes a port other than $p_i$ when located at node $v_i$ in $G$. 
The \emph{cost at node $v_i$}, denoted by $cost_{i,G}$, is $2miss_{i,G} + 1$, since there are two edge traversals for each miss and one successful edge traversal. This implies the following fact.

\begin{fact}\label{cost}
For any $G \in \mathcal{T}(D,k)$, the total cost of any treasure hunt algorithm in $G$ is $\sum_{i=0}^{D-1} cost_{i,G} = D + 2\sum_{i=0}^{D-1} miss_{i,G}$.
\end{fact}

We now prove a lower bound on the size of advice needed to solve treasure hunt for the class of caterpillar trees.



\begin{theorem}\label{tree-lb}
Let $D \leq C \leq e = n-1$. There exists a tree of size $\Theta(n)$, and a position of the treasure at distance $D$ from the initial position of the agent, such that treasure hunt at cost $C$ requires $\Omega(D\log \frac{e}{C})$ bits of advice.
\end{theorem}
\begin{proof}
Consider any algorithm $A$ that solves treasure hunt at cost at most $C$ using $b$ bits of advice. Let $k = \lceil n/D \rceil$.

Let $S$ be a set of maximum size consisting of trees from $\mathcal{T}(D,k)$ such that, for all trees in $S$, the agent is given the same advice string. By the Pigeonhole Principle, it follows that $|S| \geq \frac{|\mathcal{T}(D,k)|}{2^{b}} = \frac{k^D}{2^{b}}$. We proceed to find an upper bound on the size of such a set $S$.

Consider any two different trees $G,G' \in \mathcal{T}(D,k)$ such that the agent is given the same advice string for both of them. Let $i$ be the smallest index such that
the port at $v_i$ leading to $v_{i+1}$ is different in $G$ and $G'$. Then the behaviour of the agent prior to visiting $v_i$ for the first time is the same in $G$ and in $G'$.
Hence, $miss_{i,G'} \neq miss_{i,G}$. 
By Fact \ref{cost}, we know that $C \geq D + 2\sum_{i=0}^{D-1} miss_{i,G}$, so $\sum_{i=0}^{D-1} miss_{i,G} \leq (C - D)/2$. 
Therefore, the number of trees in $S$ is bounded above by the number of distinct integer-valued $D$-tuples of non-negative terms whose sum is at most $(C-D)/2$. (These tuples correspond to sequences $(miss_{0,G}, \dots ,miss_{D-1,G}$).)

If $(C-D)/2 < 1$, then there is only one such $D$-tuple, i.e., the tuple with all entries equal to 0. It follows that $|S|=1$. Recall that $S$ was chosen as a set of maximum size such that, for all trees in $S$, the same advice is given to the agent. It follows that, for each tree in $\mathcal{T}(D,k)$, the agent is given a different advice string. Therefore, the number of different advice strings is $k^D$, so the size of advice is at least $\log(k^D) = D\log{k} = D\log\lceil n/D \rceil$. Since $C \geq D$, and $(C-D)/2 < 1$ implies that $C < D+2$, it follows that $D\log\lceil n/D \rceil \in \Omega(D\log\frac{e}{C})$, as required.

So, we proceed with the assumption that $(C-D)/2 \geq 1$. The following claim will be used to obtain an upper bound on the number of distinct integer-valued $D$-tuples of non-negative terms whose sum is at most $(C-D)/2$. In the sequel, $D$-tuples with integer coordinates will be called {\em integer points}.

\begin{claim}\label{tuples}
Fix any $M,D \geq 1$. Let $P$ be the set of integer-valued $D$-tuples of non-negative terms whose sum is at most $M$. Then, $|P| \leq \frac{(6M)^D}{D!}$.
\end{claim}
To prove the claim, we note that $|P|$ is the number of integer points in the simplex $X = \{(x_0,\ldots,x_{D-1}) \in \mathbb{R}^D\ |\ \sum\limits_{i=0}^{D-1} x_i \leq M \textrm{ and } 0 \leq x_i \leq M \textrm{ for all $i \in \{0,\ldots,D-1\}$}\}$. Let $X_M$ denote the simplex $\{(M+x_0,\ldots,M+x_{D-1}) \in \mathbb{R}^D\ |\ \sum\limits_{i=0}^{D-1} x_i \leq (D+1)M \textrm{ and } 0 \leq x_i \leq M \textrm{ for all $i \in \{0,\ldots,D-1\}$}\}$. Since $X_M$ is a translation of the points in $X$ by $M$ in every coordinate, it follows that $|P|$ is also the number of integer points in the simplex $X_M$. For each integer point $p$ in $X_M$, we construct a small $D$-dimensional {\em box} centered at $p$. More specifically, for each $p = (p_0,\ldots,p_{D-1}) \in X_M$ such that $p_0,\ldots,p_{D-1} \in \mathbb{Z}$, we construct $B_p = \{(p_0 + \alpha_0,\ldots,p_{D-1}+\alpha_{D-1})\ |\ -1/4 \leq \alpha_i \leq 1/4 \textrm{ for each $i \in \{0,\ldots,D-1\}$}\}$. Note that, for any two distinct integer points $p,p' \in X_M$, the boxes $B_p$ and $B_{p'}$ are disjoint. Further, the volume of each such $B_p$ is $(1/2)^D$. Finally, we wish to find an upper bound on the volume of the union of all boxes $B_p$ where $p$ is an integer point in $X_M$.  To this end, we define a simplex $Y$ (a scaled version of $X$) such that, for each integer point $p \in X_M$, the box $B_p$ is completely contained in $Y$. In particular, we define $Y = \{(y_0,\ldots,y_{D-1}) \in \mathbb{R}^D\ |\ \sum\limits_{i=0}^{D-1} y_i \leq 3M \textrm{ and } 0 \leq y_i \leq 3M \textrm{ for all $i \in \{0,\ldots,D-1\}$}\}$. It follows that $|P| \cdot (1/2)^D$ is bounded above by the volume of $Y$. From \cite{ellis}, the volume of $Y$ is equal to $\frac{(3M)^D}{D!}$, which implies that $|P| \leq \frac{(6M)^D}{D!}$. This completes the proof of the claim.


By Claim \ref{tuples} with $M=\frac{C-D}{2}$, the number of trees in $S$ is bounded above by $\frac{3^D(C-D)^D}{D!}$. Combined with our earlier lower bound on the number of trees in $S$, we have $\frac{k^D}{2^{b}} \leq |S| \leq \frac{3^D(C-D)^D}{D!}$, which implies that

\[
2^{\frac{b}{D}} \geq \frac{k \cdot \sqrt[D]{D!}}{3(C-D)} \geq \frac{k \cdot \sqrt[D]{D!}}{3C}.
\]
So,
\[
b \geq D\log \left( \frac{k \cdot \sqrt[D]{D!}}{3C} \right).
\]

By Stirling's formula we have $D ! \geq \sqrt{D}(D/e)^{D}$, for sufficiently large $D$. Hence $\sqrt[D]{D !}\geq  D^{1/(2D)} \cdot (D/e)$, where $e$ is the Euler's constant.
Since the first factor converges to 1 as $D$ grows, we have $\sqrt[D]{D !} \in \Omega(D)$. Hence, the above bound on $b$ implies $b \in \Omega\left(D\log\frac{Dk}{C}\right)$. Since $k = \lceil n/D \rceil$, it follows that $b \in \Omega\left(D\log\frac{n}{C}\right)$, so the size of advice is in $\Omega\left(D\log\frac{e}{C}\right)$, as required.
\end{proof}


\section{Conclusion}

We established upper and lower bounds on the minimum size of advice sufficient to solve the problems of rendezvous and of treasure hunt at a given cost.
For the class of trees our bounds are almost tight, up to constant factors and a summand of $O(\log\log n)$. For the class of arbitrary graphs, our bounds  leave a gap of a logarithmic factor. Closing these gaps is a natural open problem. It should be noted, however, that, even for arbitrary
graphs, our bounds are asymptotically tight whenever $D\log \frac{e}{C}$ is $\Omega(\log\log e)$ and $D$ is polynomial in the gain $\frac{e}{C}$. This is the case, for example, when we want to accomplish treasure hunt
or rendezvous at cost $\Theta(\sqrt{n})$ in an $n$-node graph. 
There are only two special situations when our gap for arbitrary graphs remains non-constant. One of them is 
if $D$ is very large with respect to the gain $\frac{e}{C}$, e.g., for an $n$-node graph with $\Theta(n^{3/2})$ edges in which the treasure is located at distance $\Theta(\sqrt{n})$  at cost $\Theta(n^{3/2}/\log n)$; our (multiplicative) gap 
is $\Theta(\log n/\log\log n)$ in this case. The other situation is when both $D$ and $\frac{e}{C}$ are very small with respect to $e$, e.g., when the treasure
in an $n$-node graph
is located at distance $D \in O(\log\log\log n)$ and we want to do treasure hunt at cost $\Theta(n/\log\log n)$. In this case we have an additive gap of 
$\Theta(\log\log n)$.    

It should also be noted that, in the context of advice, treasure hunt is not only equivalent to rendezvous of two agents,
as shown in Proposition \ref{eq}, but also to rendezvous of many agents, which is often called {\em gathering}. This task consists in gathering several agents at the same node
in the same round.
In this case, the cost should be defined as the maximum number of edge traversals per agent,
and the advice size as the maximum number of bits per agent. The reduction given by the first part of Proposition \ref{eq} should be modified
as follows. One of the agents, starting at some node $w$, is given advice string $(0)$ indicating that it should be inert. Each other agent $j$ is given the advice string
$(1\alpha_j)$, where $\alpha_j$ is the advice enabling agent $j$ to find a treasure located at $w$.

\section*{Acknowledgements}
This work was partially supported by NSERC discovery grant 8136 -- 2013 and by the Research Chair in Distributed Computing at the Universit\'e du Qu\'{e}bec en Outaouais.


\end{document}